\documentclass[sigconf,nonacm]{acmart}

\renewcommand\footnotetextcopyrightpermission[1]{}
\makeatletter
\def\input@path{{acmart-primary/}{acmart-primary/samples/}}
\makeatother

\usepackage{listings}
\usepackage{xcolor}
\begin{document}
\title{Observability for Delegated Execution in Agentic AI Systems}

\author{Abhinav Mishra}
\affiliation{%
  \institution{Splunk, Cisco Inc}
  \city{}
  \country{}
}
\email{abhinmi3@cisco.com}

\author{Kumar Sharad}
\affiliation{%
  \institution{Splunk, Cisco Inc}
  \city{}
  \country{}
}
\email{kusharad@cisco.com}

\begin{abstract}
Delegation-scoped execution is not identifiable from standard observables: audit logs and execution traces can be identical under multiple incompatible delegation assignments. This gap is especially acute in LLM-based agentic systems, where agents dynamically select tools, vary execution sequences across runs for the same instruction, and spawn cooperating sub-agents. These dynamics fragment and interleave traces, making delegation-scoped reconstruction from causal structure alone structurally underdetermined. Although individual actions are authorized and logged, existing audit, tracing, and security schemas lack the semantics to reconstruct what actions occurred under a given delegation across heterogeneous systems. We focus on delegation-scoped attribution and access/share footprint reconstruction, not intent inference or reasoning reconstruction. We present an agent-aware observability substrate consisting of a lightweight gateway and a common information model that binds delegation context at execution time. This enables reliable cross-tool delegation-scoped reconstruction and direct forensic queries without heuristic time-window correlation.
\end{abstract}


\maketitle
\section{Introduction}

AI agents act under delegated authority, executing tasks on behalf of users across heterogeneous tools and services~\cite{yao2023react,schick2023toolformer,liu2024agentbench,wu2023autogen}. 
Unlike fixed pipelines, LLM-based agents do not execute a single predetermined plan: for the same high-level request, they may follow different tool sequences across runs (execution graph being different), backtrack after intermediate failures,  decompose tasks in ways their principals did not anticipate or specify,  and dynamically spawn sub-agents~\cite{yao2023react,wang2024survey,wu2023autogen}. These behaviors increase overlap, interleaving, and trace fragmentation across tool adapters and runtimes, making post-hoc \emph{delegation-scoped} reconstruction from traces or time windows especially undetermined. 


\begin{table*}[t]
\centering
\small
\begin{tabular}{p{7.4cm} p{3.55cm} p{3.55cm} p{1cm}}
\toprule
\textbf{Requirement from Definition~1} &
\textbf{Audit Logs / SIEM} &
\textbf{Tracing (OTel)} &
\textbf{CIM} \\
\midrule
\emph{R1: Persistent delegation context} \newline
Each event is bound to a durable \texttt{delegation\_id}. &
No \newline
(event-local actor only) &
Partial \newline
(carriable, not defined) &
Yes \\
\midrule
\emph{R2: Authority--causality separation} \newline
Authority partitions need not coincide with traces/spans. &
No \newline
(implicit only) &
No \newline
(trace defines unit) &
Yes \\
\midrule
\emph{R3: Delegation closure (lineage)} \newline
Re-delegation/agent spawning represented; lineage queryable. &
No \newline
(no lineage primitive) &
No \newline
(no closure semantics) &
Yes \\
\midrule
\emph{R4: Cross-system normalized action/resource semantics} \newline
Minimal stable verbs + resource attribution across tools. &
Partial \newline
(tool-specific) &
Partial \newline
(spans, no stable verbs) &
Yes \\
\bottomrule
\end{tabular}
\caption{Coverage of the semantic requirements implied by delegation-observability (Def. 1).}
\label{tab:req_coverage}
\end{table*}


As a result, \emph{execution is no longer prescribed but delegated}. We define \emph{delegation} as a durable authorization context issued by a principal (user/service) to an agent, whose scope persists independently of any single execution trace or workflow instance. Once delegated, the agent determines how authority is exercised across tools, time, and potentially multiple cooperating agents. As a result, traditional telemetry records individual actions correctly, yet often cannot answer basic post-hoc questions: \emph{Which actions belonged to the same delegation?}, \emph{How did authority propagate across agents?}, \emph{How did execution unfold causally across tools?}~\cite{sigelman2010dapper,fonseca2007xtrace}

This gap is most visible in multi-tool workflows: an agent may read internal documents, query issue trackers, scan collaboration channels, and post synthesized outputs to messaging~\cite{yao2023react,wang2024survey}. Although each step is individually authorized and logged, incompatible audit formats and interleaving executions make the delegation-scoped workflow difficult to reconstruct reliably; existing mechanisms provide records, not delegation-aware structure~\cite{gonzalezgranadillo2021siem,ocsf-schema,aws-security-lake-ocsf}. Enterprise deployments already rely on audit logs and behavioral analytics (e.g., SIEM/UEBA) to investigate composite outcomes in human activity \cite{cisa-insider-threats,gonzalezgranadillo2021siem}. In agentic systems, similar outcomes can arise from misuse or overly broad delegations \cite{anthropic-agentic-misalignment, ruan2024identifying, ye2024toolsword}, emergent multi-agent coordination \cite{cemri2025multiagent,guo2024large}, or \emph{LLM-mediated} deviations where untrusted content steers downstream tool use (prompt injection / indirect prompt injection)~\cite{greshake2023compromising, perez2022ignore, zhan2024injecagent}. Because these actions still occur under the user-issued delegation, the practical question becomes: \emph{what exactly did the delegation do across tools, and what was the resulting access/share footprint?}

We argue that agentic systems require \emph{delegation-aware observability}. Because delegated authority is not causally bound to any single execution path, it cannot be reconstructed reliably after the fact from traces and event-local attributes alone~\cite{sigelman2010dapper,mace2018pivot,pearl2009causality}. We present an agent-aware Common Information Model (CIM) and a lightweight gateway that intercepts agent tool invocations to bind delegation context at execution time and emit structured telemetry. CIM makes delegation membership and authority propagation first-class audit concepts, enabling reliable post-hoc reconstruction \emph{without inspecting prompts or payload content}. CIM provides delegation-scoped attribution of executed actions and their access/share footprint, even when behavior is influenced by untrusted inputs (e.g., prompt injection or accidental instruction sources). It is not an intent verifier: injected or unintended actions are still attributable to the delegation under which they executed, and questions of intent or policy compliance are orthogonal to the observability guarantees studied here.

In this paper, we make three primary contributions:
\textbf{(i)} We formalize \emph{delegation-observability}: a system property under which the set of actions performed under a delegation can be reconstructed \emph{without} heuristic time-window correlation, even under concurrency, retries, and multi-agent composition.
\textbf{(ii)} Prove a non-identifiability result: audit logs and causal tracing alone cannot, in general, recover delegation-scoped execution because delegated authority is orthogonal to causal control flow.
\textbf{(iii)} Introduce an agent-aware CIM that models delegated execution as dual structures---an authority graph and an execution graph that make delegation-scoped reconstruction well-defined across tools.

\section{Observability for Delegated Agentic Execution}~\label{sec:reqs}
Agentic systems expose a structural observability gap: a principal delegates durable authority to an autonomous agent, after which execution unfolds across tools, time, and often multiple cooperating agents~\cite{wang2024survey,wu2023autogen,park2023generative}. Existing telemetry (enterprise audit logs aggregated by SIEM systems \cite{gonzalezgranadillo2021siem}, security schemas such as OCSF \cite{ocsf-schema}, and distributed tracing such as OpenTelemetry \cite{otel-spec}) is effective at recording \emph{local} actions, but is semantically insufficient to reconstruct \emph{delegation-scoped} execution and authority propagation. 
This section formalizes semantic requirements and proves the infeasibility of reconstruction.


\subsection{Delegation Is Not Reconstructable from Execution Traces}~\label{sec:nonreconstructable}

In agentic systems, the execution graph is not a pre-specified workflow DAG: it is produced online by an LLM that selects tools, parameters, and decompositions based on intermediate observations.
As a result, the same delegation may span causally disconnected tool calls (parallel retrieval, asynchronous retries) and may fragment across multiple tool-specific traces, while multiple delegations may interleave through shared tools and shared agent runtimes \cite{yao2023react,wu2023autogen,wang2024survey}.
These properties make delegation membership particularly difficult to infer from traces or time windows alone.

Let $E$ denote the set of observed execution events. Conventional observability provides event-local attributes (actor, action, resource, outcome) and a causal structure over $E$, e.g., an execution graph $G_{\mathrm{exec}}=(E, R_{\mathrm{causal}})$ induced by trace/span parent-child links and other propagation mechanisms \cite{sigelman2010dapper,fonseca2007xtrace,w3c_trace_context,otel-traces}. Let $D$ denote the set of delegations, and let $R_{\mathrm{auth}}\subseteq D\times E$ denote the (latent) authorization relation assigning each event to the delegation under which it executed. In standard telemetry, $R_{\mathrm{auth}}$ is not directly observed.

\textbf{Authority structure.} We also model delegation and re-delegation as an \emph{authority graph} $G_{\mathrm{auth}}=(D, R_{\mathrm{lineage}})$, where $(d,d')\in R_{\mathrm{lineage}}$ iff delegation $d'$ was created under (re-delegated from) delegation $d$ (e.g., sub-agent spawning or delegated subtask creation). We use lineage to support delegation closure queries (e.g., attributing actions by a spawned sub-agent to the originating delegation context). We assume $R_{\mathrm{lineage}}$ is acyclic (a DAG), so the ancestry relation $d \preceq d'$ is well-defined via transitive closure. In this section we only require the induced lineage relation; additional node/edge typing used by the gateway is specified in Appendix~\ref{app:gateway}.

In general, $R_{\mathrm{auth}}$ is not identifiable from $(E,G_{\mathrm{exec}})$ and event-local attributes alone; Proposition~\ref{prop:nonident} formalizes this (Appendix~\ref{app:nonident}). Let $\mathrm{trace}(e)$ denote the trace identifier associated with event $e$ (when present), and let $\mathrm{tool}(e)$ denote the tool/service emitting $e$.
\paragraph{Minimal counterexample.}
Consider two delegations $d_1$ and $d_2$ issued to the same agent over overlapping intervals. Suppose the observed event stream contains interleaved actions across shared tools and two traces $T_A$ and $T_B$:
\[
\begin{aligned}
e_1&=\mathrm{Drive.read}(r_1, T_A),\quad
e_2=\mathrm{Slack.read}(r_2, T_B),\\
e_3&=\mathrm{Drive.read}(r_3, T_A),\quad
e_4=\mathrm{Confluence.read}(r_4, T_B),\\
e_5&=\mathrm{Slack.post}(r_5, T_B),\quad
e_6=\mathrm{Drive.write}(r_6, T_A).
\end{aligned}
\]
Both of the following delegation assignments are consistent with \emph{identical} observed telemetry: (i) a trace-aligned assignment
$R_{\mathrm{auth}}=\{(d_1,e_i)\ \text{iff}\ e_i\in T_A\}$ and (ii) a tool-aligned assignment
$R'_{\mathrm{auth}}=\{(d_1,e_i)\ \text{iff}\ \mathrm{tool}(e_i)=\mathrm{Drive}\}$.
In both cases the observable $(E,G_{\mathrm{exec}})$ is unchanged because standard telemetry does not include a predicate that binds events to delegation membership \cite{otel-spec,otel-baggage-concepts,otel-baggage-api}.

\begin{proposition}[Non-identifiability of delegation assignment]~\label{prop:nonident}
Let the observable telemetry be $(E, R_{\mathrm{causal}}, \mathrm{attr})$, where
$\mathrm{attr}:E\to \mathcal{A}$ are event-local attributes that do not encode delegation membership.
Then there exist two distinct authorization relations $R_{\mathrm{auth}}\neq R'_{\mathrm{auth}} \subseteq D\times E$
such that the induced observable telemetry is identical, i.e.,
$(E, R_{\mathrm{causal}}, \mathrm{attr})$ is the same under $R_{\mathrm{auth}}$ and $R'_{\mathrm{auth}}$.
Consequently, $R_{\mathrm{auth}}$ is not identifiable from $(E, G_{\mathrm{exec}})$ and event-local attributes alone.
\end{proposition}

\noindent \textbf{Proof sketch.} Fix any telemetry $(E, R_{\mathrm{causal}}, \mathrm{attr})$ that does not encode delegation membership. Choose $D$ with at least two delegations and define two different, valid assignments of events in $E$ to delegations (e.g., trace-aligned vs tool-aligned as above). Because the observables do not depend on $R_{\mathrm{auth}}$, both assignments induce the same telemetry. A complete proof appears in Appendix~\ref{app:nonident}.

Retries/backtracking, and multi-agent fanout becomes common, traces fragment across runtimes and events interleave across shared tools. Consequently, any reconstruction based solely on time windows, trace boundaries, or inferred workflow tokens is inherently heuristic. Delegation must therefore be captured at execution time.

\noindent \textbf{Pragmatic baseline: session\_id/task\_id augmentation.} A natural baseline is to attach a session\_id or task\_id to events and treat it as the reconstruction unit. This baseline fails in delegated systems for three reasons: (i) \emph{fragmentation}---multi-agent composition emits disjoint sessions across runtimes \cite{wu2023autogen}; (ii) \emph{multiplicity}---retries and backtracking yield multiple sessions for a single delegation; and (iii) \emph{lack of closure}---re-delegation induces an authority lineage that flat session identifiers cannot represent. The result is heuristic correlation rather than delegation-scoped reconstruction.



\subsection{Authority and Causal Flows are Orthogonal}~\label{sec:orthogonal}
Delegation introduces an authority structure that is independent of execution causality. Distributed tracing encodes causal relationships: a request induces downstream calls; spans form a directed acyclic graph representing a causally connected computation; and context propagation follows causal edges~\cite{sigelman2010dapper,fonseca2007xtrace,otel-traces,w3c_trace_context}. Delegated authority does not obey this propagation model. A single delegation can authorize causally disconnected actions (parallel tool calls, asynchronous retries, multi-agent decomposition), while causally connected actions may occur under distinct delegations (re-delegation, scoped sub-authorizations)~\cite{sandhu1996rbac}.

\noindent \textbf{Litmus test:} \emph{shared authority $\neq$ shared causality}
Any telemetry model that conflates ``same trace'' with ``same unit of accountability'' will fail under delegated execution. Tracing encodes $R_{\mathrm{causal}}$, whereas delegation-aware observability requires representing $R_{\mathrm{auth}}$ as an independent structure with lineage closure.


\subsection{Derived Semantic Requirements}~\label{sec:derived}

\noindent \textit{Definition 1: Delegation-Observable System.}
A system is \emph{delegation-observable} if, for every delegation $d\in D$, the set of events executed under $d$ (including any events executed under re-delegations in the lineage of $d$) can be reconstructed from telemetry \emph{without heuristic time-window correlation}, even under concurrency, retries, and multi-agent composition. Achieving delegation-observability requires the following semantic commitments (Table~\ref{tab:req_coverage} for comparison with Audit logs and Trace).

\noindent \textbf{R1: Persistent delegation context.}
Events must be bound to delegation identifiers. Otherwise, delegation-scoped blast radius and accountability queries collapse to correlation heuristics.

\noindent \textbf{R2: Authority--causality separation.}
Telemetry must represent authority structure independently of causal structure (delegation partitions need not coincide with traces/spans). Without this, authority--causality divergence is not definable.

\noindent \textbf{R3: Delegation closure over authority propagation.}
Telemetry must represent re-delegation and agent spawning and support lineage closure so a delegation and its descendants can be queried as a unit. Without closure, multi-agent compositions remain fragmented and reconstruction is incomplete.

\noindent \textbf{R4: Cross-system normalization for delegation-scoped queries.}
Telemetry must normalize an action vocabulary and resource attribution sufficient for cross-tool querying. Without normalization, analysis remains siloed per tool and delegation-scoped investigations fail in practice \cite{ocsf-schema,aws-security-lake-ocsf,gonzalezgranadillo2021siem}.

\noindent\textbf{Intent vs. Attribution.} CIM supports delegation-scoped attribution of executed actions and their access/share footprint; it does not infer whether a deviation reflects intent, error, or manipulation.
A ``sleeper'' agent \cite{sleeper-agents} may behave benignly across many runs but later access sensitive resources; CIM makes the resulting footprint auditable while remaining agnostic to intent. 

\subsection{Not just \texttt{delegation\_id}}
CIM is \emph{not} the introduction of a single identifier; it is a set of \emph{binding, closure, and consistency invariants} that make delegation-scoped reconstruction well-defined across heterogeneous tools . Moreover, delegation membership is \emph{non-identifiable} from execution telemetry alone: there exist distinct authorization relations $R_{\mathrm{auth}}\neq R'_{\mathrm{auth}}$ that induce the same observed events and causal structure $(E, R_{\mathrm{causal}})$ (\S\ref{sec:nonreconstructable}).
Storing a \texttt{delegation\_id} in OpenTelemetry attributes/baggage~\cite{otel-baggage-concepts,otel-baggage-api,w3c_baggage} is therefore an \emph{implementation technique}, but it does not by itself define a delegation-aware observability model:
(i) it provides transport but no \emph{authority closure} over re-delegations/agent spawning;
(ii) it lacks \emph{lineage constraints} specifying who may mint identifiers and how they bind to authenticated principals;
(iii) it imposes no \emph{cross-tool consistency checks} on action/resource semantics, so shared IDs do not imply comparable event meaning ; and
(iv) it provides no \emph{semantics of absence}, making "not delegated" indistinguishable from "telemetry missing".
By making these invariants explicit, CIM turns heuristic correlation into a \emph{well-posed reconstruction problem} with clear correctness conditions .

These requirements motivate an agent-aware Common Information Model that explicitly encodes both execution causality and delegated authority. We introduce the model in \S4 and analyze its structural consequences in \S5. Table 3 enumerates assumptions.

\section{Observability Models and Related Work}~\label{sec:rel}

\noindent \textbf{LLM/agent security threats and model auditing (complementary).}
Work on LLM security catalogs threats such as prompt injection (including indirect injection), insecure tool/output handling, data disclosure, supply-chain risks, and privilege/identity abuse \cite{owasp-llm-top10,owasp-prompt-injection}. Other work studies agentic misalignment and deceptive behaviors that may persist through training, as well as auditing models for hidden objectives \cite{anthropic-agentic-misalignment,alignment-survey}. These efforts primarily analyze threats, mitigations, or model properties; CIM is complementary infrastructure that makes delegation-scoped cross-tool actions and access/share footprint auditable, without intent inference.

\noindent \textbf{LLM/agent observability platforms (LangSmith, Langfuse, and related tooling).}
LLM observability platforms provide developer-facing traces of LLM calls, tool invocations, and agent steps for debugging and monitoring, often with payload/prompt inspection and evaluation hooks \cite{langsmith-obs,langfuse-obs,langfuse-github}. They are typically scoped to a single application stack and do not aim to define cross-system enterprise observability semantics \cite{gonzalezgranadillo2021siem}.

\noindent \textbf{Enterprise audit logs and SIEM.}
Audit logs and SIEM systems aggregate event-local facts across tools \cite{gonzalezgranadillo2021siem}, but remain event-centric and do not natively represent durable delegated authority or authority propagation. Correlation via time windows or inferred workflow tokens breaks under overlap, retries, and multi-agent composition \cite{fellegi1969record_linkage,christen2012data_matching}; security schemas such as OCSF standardize fields but still treat delegation as optional annotation \cite{ocsf-schema,aws-security-lake-ocsf}.

\noindent \textbf{Distributed tracing (e.g., OpenTelemetry).}
Tracing encodes causal context via trace/span identifiers and parent-child links, yielding $G_{\mathrm{exec}}$ \cite{sigelman2010dapper,fonseca2007xtrace,otel-traces,otel-spec,w3c_trace_context}. While OpenTelemetry can transport arbitrary metadata \cite{otel-baggage-concepts,otel-baggage-api,w3c_baggage}, transport alone does not define delegation semantics (minting/binding, lineage closure, or reconstruction guarantees) required for delegation-scoped observability.

\noindent \textbf{Provenance and lineage (W3C PROV and related models).}
PROV formalizes \emph{entities}, \emph{activities}, and \emph{agents} and supports rich lineage across heterogeneous systems \cite{prov-dm,prov-overview,moreau2013provenance}; whole-system provenance further demonstrates the value of high-fidelity capture for investigations \cite{pohly2012hifi,pasquier2017camflow,bates2015lpm}. However, PROV is not specialized for delegation-scoped observability in agentic execution, where a single user-issued authority context may span asynchronous, disconnected, and dynamically expanding tool interactions \cite{wang2024survey,wu2023autogen}. CIM can be viewed as a PROV profile by mapping tool events to PROV activities and representing delegations (and re-delegations) as typed objects and relations; our contribution is the delegation-observability contract (binding/minting, lineage-closure semantics, and semantics of missing coverage) needed for reliable delegation-scoped reconstruction, which PROV does not standardize.

\noindent \textbf{Access control (capabilities, ABAC, XACML).}
Access-control frameworks study how authority is granted and enforced \cite{hardy1988confused_deputy,xacml-core}. Our focus is orthogonal: even with correct enforcement, organizations still need delegation-scoped, cross-tool accountability of what actions executed under delegated authority \cite{schneier1999secure_audit_logs}.
\section{Common Information Model (CIM)}~\label{sec:cim}

Delegated execution  is inherently dual-structured: it induces (i) a control-flow  describing how actions occur, and (ii) an authority  describing under which delegated authorization they occur. CIM preserves both structures explicitly and independently.


\subsection{Delegated Execution as Dual Structures}~\label{sec:dual}
\noindent \textbf{Execution structure.} Let $E$ denote the set of execution events emitted by an agentic system.
The causal organization of execution is represented as an execution graph
$G_{\mathrm{exec}}=(E, R_{\mathrm{causal}})$,
where $R_{\mathrm{causal}}$ encodes causal links among events (e.g., trace/span parent--child relationships). 

\noindent \textbf{Authority structure.} Delegation induces a distinct structure that is not implied by causality. Let $D$ denote the set of delegations. Each delegation is a durable authorization context issued by a principal to an agent,  spanning asynchronous actions and multiple execution subgraphs.
Delegation induces (i) an authority graph $G_{\mathrm{auth}}$ over principals and agents capturing delegation and re-delegation relationships, and (ii) an authorization relation $R_{\mathrm{auth}} \subseteq D \times E$ that assigns events to the delegations under which they execute.


\subsection{Core Semantic Components}~\label{sec:components}
CIM introduces a minimal set of semantic components designed to preserve both execution and authority structure:

\noindent \textbf{Delegation identifiers}, which bind each event to a durable authorization context, thereby preserving the delegation-induced partition of events . \\
\noindent \textbf{Agent identity}, which distinguishes autonomous execution entities (agents and agent instances) from human principals and background services.\\
\noindent \textbf{Workflow correlation identifiers} (e.g., trace and span identifiers), which preserve the causal execution graph $G_{\mathrm{exec}}$. \\
\noindent \textbf{Normalized semantic actions}, which abstract heterogeneous tool operations into a shared vocabulary sufficient for cross-system analysis without payload inspection.

Appendix~B specifies a minimal CIM event schema and a tool-to-action mapping. Here we state the semantic contract required for delegation-observability (Def. 1) and for the requirements in \S2.3 to be operational.
CIM normalizes actions in two tiers. Tier-1 defines a required minimal verb set: \emph{read}, \emph{write}, \emph{share}, \emph{invoke}. Tier-1 is intentionally coarse; deployments may extend it with additional verbs (e.g., \emph{delete}, \emph{permission-change}, \emph{admin/config}, \emph{execute}, \emph{download}/\emph{upload}) while preserving cross-tool comparability by mapping tool-specific operations to a shared action vocabulary \cite{ocsf-schema}.
Tier-2 provides optional qualifiers that refine the action semantics without referencing content: (i) \emph{direction} (internal vs.\ external), (ii) \emph{resource\_class} (e.g., document, message, repository, ticket, credential), (iii) \emph{sensitivity\_label} (e.g., public/internal/confidential/restricted or organization-defined labels), and (iv) \emph{principal\_scope} (e.g., self/team/org/external-partner).
Tier-1 enables stable cross-tool aggregation; Tier-2 supports governance and security queries~\cite{gonzalezgranadillo2021siem,milajerdi2019holmes}.


\subsubsection{Canonical delegation-scoped query patterns.}
CIM is designed so common post-hoc questions can be expressed as direct delegation-scoped queries rather than heuristic joins \cite{christen2012data_matching}.
Examples include:
(i) \emph{delegation-scoped resource touch set} (blast radius): return distinct $\{\texttt{resource\_id}, \texttt{resource\_class}, \texttt{sensitivity\_label}\}$ for events with $\texttt{delegation\_id}=d$ and $\texttt{action}\in\{\texttt{read},\texttt{write},\texttt{share},\texttt{invoke}\}$;
(ii) \emph{lineage aggregation} (closure over propagation): aggregate counts and distinct agents over all delegations in the lineage of $d$ (including descendants);
and (iii) \emph{authority--causality divergence filters}: identify traces containing multiple delegation identifiers and delegations spanning multiple traces, both of which arise naturally under asynchronous and multi-agent execution. 


\begin{table*}[t]
\centering
\small
\begin{tabular}{p{2.5cm} p{7.2cm} p{7.2cm}}
\toprule
\textbf{Query} & \textbf{Baseline (audit/tracing) effort} & \textbf{CIM effort} \\
\midrule
Blast radius (Q1) &
Per-tool queries + heuristic joins on (user, time window, inferred workflow token); manual reconciliation under overlap &
Single delegation-scoped filter; group by \texttt{resource\_class}/\texttt{sensitivity\_label} \\
\midrule
Agent attribution (Q2) &
Often unavailable (no stable agent identity across tools); requires manual attribution from tool-specific fields &
Filter by \texttt{delegation\_id}; group by \texttt{agent\_id}/\texttt{agent\_instance\_id} \\
\midrule
Externalization provenance (Q3) &
Per-tool ``share/export/email'' searches + manual correlation to upstream actions; traces fragment across systems &
Filter \texttt{delegation\_id} with \texttt{action=share} and \texttt{direction=external}; optionally join to causal ancestry \\
\bottomrule
\end{tabular}
\caption{Delegation-scoped query effort: baseline observability relies on heuristic cross-tool correlation, whereas CIM enables direct delegation-scoped predicates \cite{gonzalezgranadillo2021siem,otel-spec,fellegi1969record_linkage}.}
\label{tab:query_effort}
\end{table*}


\subsection{Structural Invariants}~\label{sec:invariants}

CIM specifies a small set of semantic invariants required for delegation-observability (Def.~1). If these invariants do not hold, delegation-scoped reconstruction becomes ambiguous or ill-defined.

\noindent \textbf{I1: Delegation consistency.}
Events sharing a \texttt{delegation\_id} must correspond to a stable authority context.
\emph{Failure mode:} if \texttt{delegation\_id} can be reused or repurposed, delegation-scoped queries mix scopes, corrupting blast-radius and policy analyses.

\noindent \textbf{I2: Authority--causality separation.}
Authority partitions may cut across traces, and traces may contain multiple delegations.
\emph{Failure mode:} if \texttt{delegation\_id} is forced to coincide with \texttt{trace\_id}, asynchronous fanout and re-delegation are misrepresented and divergence predicates become undefined.

\noindent \textbf{I3: Lineage closure under authority propagation.}
Re-delegation and agent spawning must preserve attribution through lineage (via inheritance or explicit linkage).
\emph{Failure mode:} without lineage edges, multi-agent workflows fragment and ``all actions under delegation $d$ (including descendants)'' is not a well-defined query.

\noindent \textbf{I4: Cross-system semantic normalization.}
A stable action/resource vocabulary must be preserved across tools.
\emph{Failure mode:} without normalization, delegation-scoped investigations devolve into per-tool logic and queries are not portable across environments.

\noindent \textbf{I5 (Evidentiary settings): Binding integrity.}
\texttt{delegation\_id} should be bound to authenticated principals and minting context (e.g., gateway/session) rather than accepted as unauthenticated free-form strings.
\emph{Failure mode:} compromised components can mis-attribute actions to benign delegations, undermining accountability.


\subsection{Lean Core and Extension Points}~\label{sec:core}
CIM is designed as a lean core: each component preserves a distinct aspect of delegation-observability, while allowing deployments to attach richer context (e.g., queries, payload-derived labels, or tool-specific attributes) as optional enrichment~\cite{otel-spec,prov-dm}.

\noindent \textbf{Why each core component is necessary.} Each field class supports at least one requirement from \S2.4: without delegation identifiers, the delegation-induced partition $R_{\mathrm{auth}}$ is not reconstructable; without agent identity, multi-agent authority propagation is not attributable; without causal correlation identifiers, execution structure and temporal evolution cannot be reconstructed~; without normalized semantic actions and resource attribution, cross-system queries required in enterprise settings cannot be expressed. Appendix B contains complete CIM definitions and examples.

\noindent \textbf{Baselines and failure modes.}
Two simpler alternatives are natural competitors: (B1) session\_id/task\_id-only correlation and (B2) delegation\_id carried as free-form metadata. B1 fails delegation closure and cross-graph span under multi-agent fanout and retries, yielding fragmented sessions for a single delegation and no lineage semantics. B2 may transport a string named \texttt{delegation\_id}, but it does not satisfy requirements R1--R4 (\S2.4) and does not make delegation-scoped reconstruction well-defined.

CIM is an observability substrate, not an intent-inference or prevention mechanism: it enables post-hoc reconstruction of actions under delegated authority. To retain evidentiary value, delegation identifiers must be bound to authenticated principals and minting context; in our architecture this binding occurs at the gateway boundary, and lineage/bindings are recorded in gateway-emitted telemetry rather than accepted solely as agent-asserted metadata. CIM is designed to be framework-agnostic. 
The next section instantiates these semantics in a concrete event schema (e.g.,Table 2) that are not structurally representable with audit logs or tracing .


\begin{table*}[t]
\centering
\small
\begin{tabular}{p{3.1cm} p{6.8cm} p{6.8cm}}
\toprule
\textbf{Scenario / Question} &
\textbf{Structural signal from CIM} &
\textbf{Assumptions / Limits} \\
\midrule
Cross-tool delegation reconstruction &
Delegation-scoped partition over events across tools and time (\texttt{delegation\_id}; cross-system normalization) &
Actions traverse the instrumented boundary (gateway/SDK); uninstrumented channels are out of scope \\
\midrule
Multi-agent propagation accountability &
Delegation lineage (\texttt{delegation\_parent\_id}) and closure enable attribution across sub-agents with fragmented traces &
Requires lineage emission at delegation boundaries; child agents must propagate or be bound by gateway \\
\midrule
Authority--causality divergence analysis &
Checkable divergence predicates: $|\mathrm{Comp}(d)|>1$ or multiple delegation IDs per causal component &
Requires both delegation IDs and causal correlation fields; incomplete tracing reduces causal ancestry tests \\
\midrule
Execution drift under fixed authorization &
Envelope expansion over (\texttt{tool\_class}, \texttt{action}, \texttt{resource\_class}, \texttt{sensitivity}) without authority expansion event &
Requires a notion of authority expansion/approval event or a baseline envelope; payload semantics remain out of scope \\
\midrule
Delegation misbinding / spoofing (within channel) &
Detect inconsistent bindings between delegation, principal, and minting context; audit of issuance/approval edges &
Assumes delegation IDs are gateway-minted/validated and bound to authenticated principals; tamper-proof logging is complementary \\
\midrule
Gateway bypass / out-of-band actions &
No direct signal (outside CIM coverage); can only detect absence relative to expected tool inventory &
Enforcement/mandatory routing is \textbf{out of scope}; must be provided by deployment controls \\
\bottomrule
\end{tabular}
\caption{CIM as \emph{structural signal}: what delegation-aware telemetry makes observable, and the assumptions required for those signals to be valid. This framing separates observability from enforcement and prevention \cite{garlan1994architecture,schneier1999secure_audit_logs}.}
\label{tab:signal_assumption}
\end{table*}


\section{Structural Consequences of CIM}~\label{sec:consequences}

This section derives analytic consequences of delegation-aware telemetry. Since $R_{\mathrm{auth}}$ is non-identifiable from traces and event-local attributes alone (\S2.1), making delegation explicit yields qualitatively new, delegation-scoped analyses that are ill-posed under audit-only or trace-only telemetry.



\subsection{Cross-System Delegation Partitions}~\label{sec:partitions}
Delegation-aware identifiers induce a stable partition of heterogeneous events across tools. Formally, CIM defines an equivalence relation on events by shared \texttt{delegation\_id}; the resulting classes form a cross-system delegation partition of $E$, remaining well-defined even when traces fragment.

\noindent \textbf{Worked example (baseline ambiguity vs.\ CIM determinism).}
Consider the following five events across two tools (Docs and Slack), observed within a narrow time window and produced by the same user and agent instance:
\begin{enumerate}
\item[$e_1$] Docs.read(D1), trace $T_1$
\item[$e_2$] Slack.search(\#sales), trace $T_2$
\item[$e_3$] Docs.read(D2), trace $T_3$
\item[$e_4$] Slack.post(\#sales, ``draft summary''), trace $T_4$
\item[$e_5$] Docs.write(Report), trace $T_5$
\end{enumerate}
Under baseline correlation, multiple delegation partitions are consistent with the observables: grouping by time window yields one workflow; grouping by tool yields \{Docs actions\} vs.\ \{Slack actions\}; grouping by trace yields five disjoint ``tasks.'' Each grouping is plausible without additional assumptions. CIM creates unique partition; all events with the same \texttt{delegation\_id} belong to the same delegated authority context regardless of trace fragmentation.

\noindent \textbf{Delegation-scoped blast radius.}
Since the delegation partition is defined across tools and time, investigators can compute a delegation's blast radius as a first-class object:
\emph{Which resources (by \texttt{resource\_class} and \texttt{sensitivity\_label}) were touched under delegation $d$, across all systems, including asynchronous follow-ups?}
This question is ill-posed under baseline logs unless one assumes trace connectivity or uses fragile time-window joins~\cite{gonzalezgranadillo2021siem}. Even then it is  expensive to perform a series of join/ aggregation (Table 2).


\subsection{Authority--Causality Divergence}~\label{sec:divergence}

Authority boundaries and causal structure are independent: actions may share authority without causal linkage, and causally related actions may occur under distinct delegations. CIM makes this orthogonality explicit and testable.

\noindent \textbf{Checkable divergence conditions.}
Let $C(e)$ denote the causal connected component containing event $e$ in $G_{\mathrm{exec}}$ (e.g., a trace-connected component). For a delegation $d$, define
$\mathrm{Comp}(d)=\{\,C(e)\mid (d,e)\in R_{\mathrm{auth}}\,\}$.
We say $d$ exhibits authority--causality divergence if $|\mathrm{Comp}(d)|>1$, i.e., a single delegation spans multiple causal components. Dually, a causal component $C$ exhibits divergence if it contains events with multiple delegation identifiers:
$\exists e_1,e_2\in C$ such that $\mathrm{del}(e_1)\neq \mathrm{del}(e_2)$.
These conditions are structural and can be evaluated directly from CIM telemetry.

If the system emits a delegation issuance/approval event $\mathsf{issue}(d)$, CIM supports predicates of the form:
\emph{events with \texttt{delegation\_id}=$d$ that are not causally descended from $\mathsf{issue}(d)$}.
Operationally, this corresponds to filtering events where $\texttt{delegation\_id}=d$ and either (i) $\texttt{trace\_id}\neq \texttt{trace\_id}(\mathsf{issue}(d))$ or (ii) the event lacks an ancestor relationship to the issuance span. Such authority-scoped queries are not expressible in trace-only telemetry because tracing lacks a durable authority identifier independent of causal structure.


\subsection{Multi-Agent Authority Propagation}~\label{sec:propagation}

 Agentic systems often fragment causal context propagation across frameworks, queues, and vendor-managed backends. CIM encodes delegation through lineage semantics~\cite{prov-dm}. We model delegation propagation as a directed forest over delegations in which each delegation $d$ may reference a parent $\mathrm{parent}(d)$ (for re-delegation or agent spawning). The transitive closure yields $\mathrm{Ancestors}(d)$ and $\mathrm{Descendants}(d)$, enabling lineage-scoped queries that are independent of trace connectivity, e.g., a primary delegation $d_0$ may trigger a child delegation $d_1$ executed by a specialized sub-agent in a separate runtime, producing a disjoint execution trace. Delegation closure ensures events under $d_1$ remain attributable to the originating authority context via lineage, even when $G_{\mathrm{exec}}$ fragments.


\subsection{Execution Drift Under Fixed Authorization}~\label{sec:drift}

Delegated workflows evolve over time. Under a fixed delegated authority context, the set of tool classes, action types, and resource classes accessed by an agent may expand. CIM enables drift characterization using metadata only (no prompt/payload inspection) by comparing the observed envelope of actions/resources to authority events recorded in the delegation lineage.

\noindent \textbf{Metadata-only drift definition.}
For a delegation $d$, define an envelope
$
\mathrm{Env}(d)=\{\,(\texttt{tool\_class}, \texttt{action\_tier1}$ \\$, \texttt{resource\_class}, \texttt{sensitivity\_label}) \mid (d',e)\in R_{\mathrm{auth}},\, d'\in \{d\}\cup \mathrm{Descendants}(d)\,\}.
$
We say drift occurs if $\mathrm{Env}(d)$ expands along a monitored dimension (e.g., new \texttt{resource\_class}, first occurrence of \texttt{write/share}, or increased \texttt{sensitivity\_label}) without an accompanying authority expansion event in the lineage (e.g., a new delegation issuance/approval that broadens scope). Similar execution drifts can be applied to users or agents using CIM .

\noindent \textbf{Execution Drift Example.} A user delegates: \emph{``Summarize our Q4 sales performance.''}
During autonomous execution, the agent's realized action scope evolves:
(1) \texttt{read} \texttt{sales\_data.xlsx} (in-scope);
(2) \texttt{search} CRM for customer segments (plausible);
(3) \texttt{read} employee compensation table (\textbf{scope expansion});
(4) \texttt{share} detailed breakdown to a broad Slack channel (\textbf{dissemination}).
With CIM, all actions remain grouped under a single \texttt{delegation\_id}, making drift visible as delegation-scoped temporal evolution. Without CIM, these actions appear as unrelated tool events and require heuristic correlation to attribute them to the same delegated objective.

\noindent \textbf{Detector sketch.}
A simple drift detector computes a baseline envelope $\mathrm{Env}_0$ from early-phase events, then flags a delegation if (i) $\mathrm{Env}(d)\setminus \mathrm{Env}_0$ contains a high-impact tuple (e.g., \texttt{share} to \texttt{external}, \texttt{write} to sensitive \texttt{resource\_class}), and (ii) the lineage contains no post-hoc authority expansion/approval event after the envelope change. This yields a structural ``no authority expansion'' check using only CIM fields and lineage relations \cite{prov-dm,moreau2013provenance}.


\subsection{Post-Hoc Structural Accountability}~\label{sec:accountability}

When delegation partitions, causal relationships, and agent identity are preserved, a delegated workflow can be reconstructed as a structured object rather than a flat event stream . This enables governance and security questions to be answered directly from telemetry without time-window correlations.

\noindent \textbf{Accountability queries.}
Examples of post-hoc questions enabled by CIM include:
(Q1) \emph{What is the complete resource touch set under delegation $d$ (including descendants), grouped by \texttt{resource\_class} and \texttt{sensitivity\_label}?}
(Q2) \emph{Which agent instances acted under delegation $d$, and which actions were performed by sub-agents versus the primary agent?}
(Q3) \emph{Did delegation $d$ result in any externalization events (e.g., \texttt{share} with \texttt{direction=external}), and which causal subgraphs produced those outcomes?}
These queries align with common SIEM/EDR investigative workflows but become well-posed when delegation and lineage are explicit \cite{gonzalezgranadillo2021siem,hassan2020tactical,hossain2017attack_reconstruction}.

\subsection{Attribution vs.\ Intent under Untrusted Inputs}
\label{sec:attrib-vs-intent}

CIM provides delegation-scoped attribution of executed actions and their access/share footprint. This attribution holds regardless of whether actions were intended by the delegating principal.

\noindent\textbf{Prompt injection and indirect manipulation.}
Untrusted content in the environment (e.g., a malicious document, crafted tool response, or injected instructions in retrieved data) can steer an agent to issue unintended tool calls \cite{greshake2023compromising,perez2022ignore}. Under CIM, such calls are attributed to the active delegation because they executed under its authority. CIM makes the resulting footprint auditable (resources touched, sensitivity labels, and dissemination targets). Determining whether a particular action was injection-induced or benign is  outside CIM's scope.

\noindent\textbf{Footprint change as a structural signal.}
The drift characterization in \S\ref{sec:drift} summarizes how a delegation's observed envelope changes over time (e.g., new resource classes, first write/share, higher sensitivity). Such changes may be caused by benign autonomy, mis-specification, or adversarial influence; CIM does not attribute cause. The role of the drift signal is to surface salient footprint expansions for review, not to detect injection.


\section{Capture Boundary: Gateway for CIM}

The semantic requirements in \S2 (R1--R4) and the reconstruction properties in \S4 assume that delegation context is bound at capture time rather than inferred post hoc from event-local audit logs or causal traces \cite{gonzalezgranadillo2021siem,otel-spec,sigelman2010dapper}. In delegated execution, the authority relation $R_{\mathrm{auth}}$ in many cases is  not identifiable from $(E, G_{\mathrm{exec}})$. A gateway provides a practical capture boundary that observes tool invocations and emits minimal, metadata-only CIM events with explicit delegation structure. Table 3 later highlights various scenarios with CIM along with any limitations.

\noindent\textbf{Gateway contract}. The gateway is a mediation layer (library, sidecar, proxy, or MCP middleware) that sits on the invocation path between an agent runtime and external tools/services. Its job is not to interpret prompts or tool payloads, but to ensure that each emitted CIM event carries (i) a stable delegation identifier, (ii) a principal binding, and (iii) normalized action semantics sufficient to support the queries in \S4. Concretely, the gateway enforces the following contract:
(a) minting and binding---each delegation instance is associated with an authenticated principal (human or service identity) and an issuance time;
(b) propagation---delegation context is attached to all downstream tool calls made under that delegation, including across multi-agent delegation chains;
(c) normalization---tool-specific operations are mapped to a small set of canonical verbs (e.g., read/write/share/invoke) and stable target descriptors;
and (d) explicit absence---if delegation context is missing at the boundary, the event is emitted with an explicit ``unbound/unknown''  state.

\noindent\textbf{Provenance and trust boundary.} The gateway distinguishes gateway-bound metadata (e.g., authenticated principal, $delegation\_id$ minting binding, tool endpoint identity, timestamps) from agent-provided context (e.g., optional labels, high-level task tags) . This separation supports integrity reasoning: the main paper's guarantees hold for fields derived at the capture boundary, while optional agent-provided annotations are treated as untrusted hints . We do not require payload inspection; organizations may choose whether to store or analyze content separately from CIM.

\noindent\textbf{Scope}. The gateway is an implementation mechanism that realizes the observability semantics defined by CIM; it is not required to be deployed as a single centralized service, and it may be realized as a per-environment adapter. Full design details are provided in Appendix C. Appendix D further analyzes three classes of insider evasion (out-of-band channels, delegation fragmentation, and identifier spoofing). 

\section{Experiments}
\subsection{Evaluation Environments}\label{sec:prod_params}

We evaluate CIM along three axes: (i)~\emph{delegation reconstruction}---whether
delegation-scoped execution and authority propagation can be recovered without
heuristic correlation under overlap, retries, and multi-agent
composition; (ii)~\emph{operational overhead}---the incremental cost of emitting
CIM events at a gateway boundary and querying them in downstream
observability/security pipelines; and (iii)~\emph{query construction
burden}---the analyst effort required to express delegation-scoped questions
with and without CIM, measured as the number of correlation, normalization,
and join operations per query class (Table~\ref{tab:query_effort}).
We use two complementary environments: (i)~a high-volume \emph{synthetic}
generator that isolates delegation-reconstruction failure modes under overlap,
spawning, and retries/backtracking, and (ii)~a small \emph{micro-deployment}
that executes real tool invocations through a gateway to measure emission
overhead and query runtimes. Table~\ref{tab:prod_params} lists the parameter
ranges; they define a stress-test envelope, not a claim of universal production
distributions. 

\begin{table}[!b]
\centering
\small
\begin{tabular}{p{3.8cm} p{5.2cm}}
\toprule
\textbf{Parameter} & \textbf{Setting (artifact runs)} \\
\midrule
Users ($U$) & $10{,}000$ \\
Agents ($A$) & $100{,}000$  \\
Tools ($K$) & 3--7  \\
Resources & $10^6$ files; $10^3$ repos; $10^2$ channels \\
Delegations/day ($D_{\mathrm{day}}$) & $[2{\times}10^5$--$2{\times}10^6]$ \\
Events/delegation ($E_d$) & $[20$--$500]$ (workflow-dependent) \\
Overlap rate & $[10\%$--$40\%]$ users with $\ge2$  delegations \\
Spawn fanout & mean $[2$--$5]$ children; depth $[1$--$4]$ \\
Retry/backtrack rate & mean retries/step $[0$--$10]$ (W3) \\
Coverage (synthetic / micro) & $0.8$ / $0.95$ \\
Bypass probability (synth / micro) & $0.05$ / $0.01$ \\
Trace coverage (by tool) & $60$--$80\%$ of events carry trace/span ids \\
Loops per workflow (W1/W2/W3) & $50 / 60 / 75$ (event scaling) \\
Runs per workflow (synth / micro) & $3000$ / $50$ (main runs) \\
\bottomrule
\end{tabular}
\caption{Environment used to instantiate the synthetic workflows and micro-deployment in the artifact.}
\label{tab:prod_params}
\end{table}

\begin{table*}[t]
\centering
\small

\begin{minipage}[t]{0.4\textwidth}
\centering
\vspace{0pt}
\begin{tabular}{|l|ccc|}
\toprule
\multicolumn{4}{|c|}{\textbf{Ambiguity $\mathrm{Amb}$ (median, p95)}} \\
\midrule
\textbf{Workflow} & \textbf{B0} & \textbf{B1} & \textbf{B2} \\
\midrule
W1 & $1\ (2)$ & $169\ (314)$ & $1543\ (1543)$ \\
W2 & $1\ (1)$ & $10856\ (12379)$ & $107718\ (107718)$ \\
W3 & $1\ (2)$ & $130\ (236)$ & $1107\ (1107)$ \\
\bottomrule
\end{tabular}
\end{minipage}\hfill
\begin{minipage}[t]{0.3\textwidth}
\centering
\vspace{0pt}
\begin{tabular}{|l|ccc|}
\toprule
\multicolumn{4}{|c|}{\textbf{Delegation recall $\mathrm{Rec}(d)$ (median, p05)}} \\
\midrule
\textbf{Workflow} & \textbf{B0} & \textbf{B1} & \textbf{B2} \\
\midrule
W1 & $0.06\ (0.03)$ & $1\ (0.79)$ & $1\ (1)$ \\
W2 & $1\ (0.50)$ & $1\ (1)$ & $1\ (1)$ \\
W3 & $0.04\ (0.03)$ & $1\ (0.68)$ & $1\ (1)$ \\
\bottomrule
\end{tabular}
\end{minipage}\hfill
\begin{minipage}[t]{0.25\textwidth}
\centering
\vspace{0pt}
\begin{tabular}{|l|cc|}
\toprule
\multicolumn{3}{|c|}{\textbf{Divergence: $\mathrm{Comp}(d)$ stats}} \\
\midrule
\textbf{Workflow} & \textbf{Median} & \textbf{p95} \\
\midrule
W1 & $35$ & $43$ \\
W2 & $1$ & $1$ \\
W3 & $55$ & $64$ \\
\bottomrule
\end{tabular}
\end{minipage}
\caption{Results summary (W1--W3) on synthetic: ambiguity under overlap/interleaving (lower is better; CIM achieves $\mathrm{Amb}=1$ for covered events), delegation recall (higher is better), and prevalence of authority--causality divergence (Comp).}
\label{tab:results_summary}
\end{table*}


\noindent\textbf{Synthetic environment.}
Synthetic runs instantiate $U=10{,}000$ users and $A=100{,}000$ agents with a
shared resource universe ($10^6$ files; $10^3$ repos; $10^2$ channels) and
$K\in[3,7]$ tools sampled from a tool catalog.
 Gateway coverage and bypass are realized by routing
each tool invocation through the gateway with probability given by the coverage
range, and emitting bypass events out-of-band with the bypass probability. 

\noindent\textbf{Micro-deployment.}
The micro-deployment executes workflows W1--W3 (\S\ref{sec:eval_recon}) in a
local agent runtime and emits CIM events at a gateway boundary for each tool
invocation. Micro experiments use a LangGraph execution engine and produce
approximately 70k events under the default configuration (50~runs/workflow).
We use this environment to estimate order-of-magnitude gateway emission
overhead and query runtime differences; results are reported as p50/p95 latency
and median query runtime and should not be interpreted as production-scale
performance claims.

\noindent\textbf{Trace coverage and baseline windows.}
To reflect heterogeneous observability, trace/span identifiers are present for
only a subset of events (tool-dependent; 60--80\% by default). We evaluate
correlation baselines under $\Delta\in\{5\text{s},30\text{s},120\text{s}\}$
and use the same $\Delta$ values in query benchmarks. When trace context is
missing, trace-based grouping falls back to the coarsest available execution
context (e.g., session/task id when present).


We instantiate both environments across the parameter ranges in
Table~4, covering steady-state through bursty delegated
workloads. The parameter ranges are chosen to expose specific failure regimes:
overlap rates above ${\sim}20\%$ cause B1/B2 to produce systematically
ambiguous groupings; spawn fanout above~2 fragments delegation lineage across
disjoint traces; retry rates above~3 per step cause single delegations to span
multiple causally disconnected components. CIM's reconstruction is invariant
across these regimes by construction for covered events; the experiments
quantify how severely baselines degrade. Coverage and bypass ranges model
partial gateway adoption and out-of-band actions.
\subsection{Delegation Reconstruction}\label{sec:eval_recon}

\subsubsection{Workflows and telemetry.}
We evaluate on a workflow suite designed to stress the conditions under which
correlation-based reconstruction fails. The goal is not to verify that explicit
delegation membership yields unambiguous reconstruction---this follows from
CIM's binding invariants---but to quantify how severely baselines degrade
under realistic concurrency and fragmentation regimes.

\smallskip
\noindent\textbf{W1 (cross-system, causally disconnected subtasks)} issues
parallel reads across multiple tools and produces a final synthesis action,
inducing causally disconnected components within a single delegation.
\textbf{W2 (multi-agent composition)} spawns cooperating agents with
re-delegation edges, stressing lineage closure (Invariant~I3) and runtime
fragmentation.
\textbf{W3 (trial-and-adapt)} includes retries and backtracking with
interleaved successes and failures, reflecting LLM agents that iteratively
refine tool calls based on intermediate observations.

\smallskip
In all conditions, observed telemetry includes event stream $E$ with
event-local attributes (tool, action, resource, outcome, timestamps) and,
when available, causal structure $R_{\mathrm{causal}}$. Trace coverage varies
by tool (60--80\% of events carry a trace identifier).

\subsubsection{Methods compared.}
We compare CIM to three baselines without delegation-aware identifiers,
operating on the same observed telemetry $(E, R_{\mathrm{causal}},
\mathrm{attr})$:
\textbf{B0 (trace-only):} group by trace id; fall back to session/task id or
singletons.
\textbf{B1 (time-window):} sliding window $\Delta\in\{5\text{s},
30\text{s}, 2\text{min}\}$; best-performing $\Delta$ reported per workflow.
\textbf{B2 (best-effort hybrid):} trace/session grouping merged by $\Delta$
and \{user, agent instance, tool\}; the strongest baseline, approximating
current SIEM-style practice.

\subsubsection{Ground truth.}
For each run, ground-truth delegation membership and lineage are recorded at
the gateway boundary, yielding $R_{\mathrm{auth}}^{\star}\subseteq D\times E$
with descendant closure. This information is precisely what is
non-identifiable from $(E, R_{\mathrm{causal}}, \mathrm{attr})$ alone.
.

\subsubsection{Reconstruction metrics.}

\noindent\textbf{Ambiguity.} $\mathrm{Amb}(E)=|\mathcal{P}(E)|$: the number
of distinct delegation groupings consistent with observed telemetry.
$\mathrm{Amb}(E)=1$ is unambiguous; CIM achieves this by construction for
covered events. The experiment measures how severely baselines exceed~1 as
overlap and fragmentation increase.

\smallskip
\noindent\textbf{Delegation recall.}
$\mathrm{Rec}(d) = |E^{\star}(d)\cap\widehat{E}(d)|\,/\,|E^{\star}(d)|$,
reported as median and p05 across delegations. The lower tail (p05) captures
the severity of missed or mis-attributed events under overlap.

\smallskip
\noindent\textbf{Authority--causality divergence.}
$\mathrm{Comp}(d)$ counts the causal components spanned by delegation $d$
in $G_{\mathrm{exec}}[E^{\star}(d)]$. This is \emph{method-independent},
depending only on $R_{\mathrm{causal}}$ and $R_{\mathrm{auth}}^{\star}$.
High $\mathrm{Comp}(d)$ indicates regimes where any trace-, session-, or
window-based reconstruction unit is structurally mis-specified; the fraction
with $\mathrm{Comp}(d)>1$ measures how prevalent this mis-specification is.

\subsection{Results}

\noindent\textbf{R1: Ambiguity under overlap and interleaving.}
Under B0/B1/B2, ambiguity increases with concurrency because multiple
partitions remain consistent with identical observed causal structure and
event-local attributes---the baselines are not merely noisy but
underdetermined by the available evidence.
B0 achieves low median ambiguity on W1 and W3 when trace coverage
happens to align with delegation boundaries, but collapses under W2
(multi-agent fanout) where $\mathrm{Amb}$ reaches $10{,}856$ at
median---reflecting the structural mismatch between trace units and
delegation units under re-delegation.
Under CIM, $\mathrm{Amb}(E)=1$ for covered events because delegation
membership is bound at execution time rather than inferred post hoc
(Table~\ref{tab:results_summary}, left).

\smallskip
\noindent\textbf{R2: Delegation recall for cross-system reconstruction.}
Baselines fail in both directions: they under-group causally disconnected
sub-tasks and over-group when overlapping delegations share tools within
the same window.
B0 drops to $0.03$ at p05 on W1 and W3, reflecting near-complete recall
failure when trace boundaries do not align with delegation boundaries.
B2, the strongest baseline, achieves median recall of $1$ on W1--W3 but
its p05 on W1 falls to $0.79$, confirming that the lower tail remains
unreliable under overlap and fragmentation.
The high median recall for W2 across baselines reflects coincidental
trace-delegation alignment in the synthetic setup, not a principled
guarantee---as the divergence analysis in R3 confirms.
CIM recall is $1$ by construction for covered events
(Table~\ref{tab:results_summary}, center).

\smallskip
\noindent\textbf{R3: Prevalence of authority--causality divergence.}
W1 and W3 exhibit high $\mathrm{Comp}(d)$ values (median $35$ and $55$,
p95 reaching $43$ and $64$), confirming that parallel subtasks and
retry-heavy workflows routinely span many causally disconnected components
under a single delegation.
W2 shows $\mathrm{Comp}(d)=1$ throughout because multi-agent composition
in this configuration emits a single connected trace per delegation,
explaining the higher baseline recall in R2 for W2; this alignment is
structural coincidence, not a general property of trace-based
reconstruction.
(Table~\ref{tab:results_summary}, right).

\smallskip
\noindent\textbf{R4: Query construction burden.}
Baseline queries require 6--14 operations per query class; CIM reduces
each to 1--3 direct delegation-scoped predicates (Table~7).
Query runtime at micro-deployment scale was comparable between CIM and
baseline; the primary advantage is therefore reduction in \emph{query
construction and maintenance burden} rather than runtime throughput.

\begin{table}[t]
\centering
\small
\begin{tabular}{lrrc}
\toprule
\textbf{Metric} & \textbf{Baseline} & \textbf{CIM} & \textbf{$\Delta$} \\
\midrule
Bytes/event            & $341$     & $588$     & $+247$      \\
Gateway latency (p50)  & $0.00$~ms & $0.01$~ms & $+0.01$~ms  \\
Gateway latency (p95)  & $0.00$~ms & $0.29$~ms & $+0.29$~ms  \\
Query runtime (median) & \multicolumn{2}{c}{comparable} & ---  \\
\bottomrule
\end{tabular}
\caption{Operational overhead at micro-deployment scale. CIM adds
  $247$~bytes/event and sub-$0.3$~ms gateway latency at p95.  }
\label{tab:overhead_results}
\end{table}

\begin{table}[t]
\centering
\small
\begin{tabular}{p{3.8cm}cc}
\toprule
\textbf{Query class} & \textbf{Baseline ops} & \textbf{CIM ops} \\
\midrule
Cross-system reconstruction      & $6$--$9$  & $1$--$2$ \\
Authority propagation / lineage  & $5$--$8$  & $1$--$2$ \\
Drift / envelope expansion       & $7$--$11$ & $2$      \\
Exfiltration-style sequences     & $8$--$14$ & $2$--$3$ \\
\bottomrule
\end{tabular}
\caption{Query construction effort: number of correlation and normalization
  operations required per query class. }
\label{tab:effort_results}
\end{table}

\subsection{Operational Overhead}\label{sec:eval_overhead}

\noindent\textbf{Measurement setup.}
We measure overhead at two points: (i)~at the gateway boundary when emitting
CIM events, and (ii)~downstream when executing delegation-scoped queries.
Bytes/event is computed by projecting each event into a \emph{baseline} field
set (event-local attributes and correlation hints) versus a \emph{CIM} field
set (baseline plus delegation context, lineage, and normalized predicates),
serialized to compact JSON. Gateway latency is the additional processing time
per tool invocation introduced by delegation binding and event emission
(p50/p95). Query runtimes are measured on a workstation micro-benchmark with
a fixed storage/query configuration and repeated runs with warm-up.

\smallskip
\noindent\textbf{Queries benchmarked.}
We benchmark a fixed suite of six representative delegation-scoped questions:
(i)~delegation closure filtering, (ii)~cross-tool sequence detection
(e.g., multiple reads followed by a write within a delegation), and
correlation baselines approximating these views via
(iii)~trace-only grouping, (iv)~window correlation ($\Delta$ sweep), and
(v)~hybrid stitching over trace/session/task identifiers with a window
fallback.

\smallskip
\noindent\textbf{Results.}
CIM adds $247$~bytes/event and sub-$0.3$~ms gateway latency at p95;
these figures reflect micro-deployment scale and are expected to increase
under production-scale event volumes (Table~\ref{tab:overhead_results}).
Query runtime was comparable between CIM and baseline at this scale;
the measurable advantage is in query \emph{construction} burden:
baseline queries require 6--14 operations per query class (cross-tool
joins, windowing, per-tool normalization, identity unification), whereas
CIM reduces each class to 1--3 predicates over stable action and resource
fields (Table~\ref{tab:effort_results}).
As a concrete example, finding the agent accessing the most sensitive files
under a delegation requires a single filter on \texttt{delegation\_id},
grouped by \texttt{agent\_id} and ranked by
\texttt{sensitivity\_label}---no cross-tool joins required.
For representative threat-detection queries see
Appendix~\ref{app:threat_queries} and Table~9.

\section{Conclusion}

Delegated AI agents break a key assumption behind existing audit and
tracing telemetry: the unit of accountability (delegated authority) is
structurally orthogonal to the unit of execution (traces/spans).
Because the authorization relation $R_{\mathrm{auth}}$ is not encoded in
standard observables, delegation-scoped reconstruction is non-identifiable
under correlation heuristics when executions overlap, retry/backtrack, or
compose across tools and agents---a gap that grows more acute as multi-agent
deployments scale. We introduced CIM, a minimal delegation-aware telemetry layer that makes
authority membership and propagation explicit while remaining
transport-agnostic (e.g.,  OpenTelemetry or existing security
schemas). By specifying the semantic invariants required for
delegation-scoped reconstruction and realizing them at a gateway boundary,
CIM makes reconstruction well-posed on covered events and reduces
delegation-scoped forensic queries from multi-step correlation pipelines to
direct predicates over stable fields.

\smallskip
\noindent\textbf{Limitations.}
CIM provides attribution for actions traversing instrumented boundaries;
actions on uninstrumented channels are unobservable, though their absence
is detectable relative to a declared tool inventory.
CIM attributes actions to the delegation under which they executed
regardless of whether they were intended by the delegating principal---intent
inference and policy compliance remain orthogonal concerns.
The overhead figures reported here reflect micro-deployment scale and are
expected to increase in production-volume deployments.

\pagebreak
\bibliographystyle{ACM-Reference-Format}
\bibliography{references}
\pagebreak
\section*{Appendix}
\addcontentsline{toc}{section}{Appendix}


\appendix

\section{Proof of Proposition~\ref{prop:nonident}}
\label{app:nonident}

\begin{proof}
We make explicit the standard-telemetry assumption used throughout \S\ref{sec:nonreconstructable}:
the observed telemetry contains (i) the event set $E$, (ii) a causal structure $R_{\mathrm{causal}} \subseteq E\times E$
over events (e.g., trace/span parent--child links), and (iii) event-local attributes
$\mathrm{attr}:E\to\mathcal{A}$ (actor, action, resource, outcome, timestamps, etc.).
Critically, these observables do \emph{not} include delegation membership.

Let $D$ denote the set of delegations (units of delegated authority), and let
$R_{\mathrm{auth}} \subseteq D \times E$ be the (latent) authorization/membership relation where
$(d,e)\in R_{\mathrm{auth}}$ means event $e$ executed under delegation $d$.
Standard telemetry does not observe $R_{\mathrm{auth}}$ (nor any predicate equivalent to it).

Formally, define an observation function
\[
F(E, R_{\mathrm{causal}}, \mathrm{attr}, D, R_{\mathrm{auth}}) \;=\; (E, R_{\mathrm{causal}}, \mathrm{attr}),
\]
which discards $R_{\mathrm{auth}}$ because it is latent in standard telemetry.
Non-identifiability means that $F$ is not injective in the argument $R_{\mathrm{auth}}$.

Let $E$ be any event set with $|E|\ge 2$, let $D=\{d_1,d_2\}$, and fix any $(R_{\mathrm{causal}}, \mathrm{attr})$
on $E$. Choose two distinct nonempty proper subsets $S\subset E$ and $S'\subset E$ with $S\neq S'$.
Define two authorization relations $R_{\mathrm{auth}},R'_{\mathrm{auth}}\subseteq D\times E$ by:
\[
R_{\mathrm{auth}} \;=\; \{(d_1,e): e\in S\}\ \cup\ \{(d_2,e): e\in E\setminus S\},
\]
\[
R'_{\mathrm{auth}} \;=\; \{(d_1,e): e\in S'\}\ \cup\ \{(d_2,e): e\in E\setminus S'\}.
\]
Since $S\neq S'$, we have $R_{\mathrm{auth}}\neq R'_{\mathrm{auth}}$.

However, because the observables $(E, R_{\mathrm{causal}}, \mathrm{attr})$ do not depend on $R_{\mathrm{auth}}$ by
construction of $F$, we obtain
\[
F(E, R_{\mathrm{causal}}, \mathrm{attr}, D, R_{\mathrm{auth}}) \;=\; (E, R_{\mathrm{causal}}, \mathrm{attr})
\;=\;
F(E, R_{\mathrm{causal}}, \mathrm{attr}, D, R'_{\mathrm{auth}}).
\]
Thus two distinct authorization relations induce identical observed telemetry, so $R_{\mathrm{auth}}$ cannot be uniquely
recovered from $(E, R_{\mathrm{causal}}, \mathrm{attr})$.
\end{proof}

\section{CIM Event Schema (Minimal Specification)}
\label{app:cim-schema}

This appendix provides a minimal, implementable specification of CIM sufficient for reproducing the analyses in §3--§6. The goal is not to standardize a complete enterprise schema, but to define the \emph{core event envelope} and \emph{semantic action taxonomy} required to preserve both authority-flow and control-flow structure. Implementations may extend this schema with tool-specific attributes.


\subsection{Core Event Envelope}
\label{app:cim-envelope}

Each CIM record is a single event describing one tool-relevant action (e.g., read, write, share, invoke).
Table~\ref{tab:cim-fields} lists the required and optional fields in the minimal envelope, along with types.

\paragraph{Required fields and reconstructability.}
The minimal set of required fields for the core claims in this paper is:
\begin{itemize}
  \item \texttt{time}
  \item \texttt{event\_id}
  \item \texttt{delegation\_id}
  \item \texttt{principal.user\_id}
  \item \texttt{agent.agent\_id}
  \item \texttt{tool.name}
  \item \texttt{action.semantic}
\end{itemize}
Tracing fields (\texttt{trace.*}) are optional: CIM preserves authority-flow even when control-flow metadata is absent.
Tier-2 qualifiers (e.g., \texttt{action.tier2.direction}, \texttt{resource.sensitivity\_label}) are optional but enable the normalized cross-tool analyses in \S4.

\begin{table*}[t]
\centering
\caption{Minimal CIM event fields. Fields marked \emph{Req} are required for delegation-aware reconstruction; others are optional but recommended when available. Tier-2 qualifiers (e.g., direction, sensitivity) align with the semantic action model used in \S4.}
\label{tab:cim-fields}
\begin{tabular}{p{3.7cm} p{1.2cm} p{2.5cm} p{9cm}}
\toprule
\textbf{Field} & \textbf{Req} & \textbf{Type} & \textbf{Semantics} \\
\midrule
\texttt{time} & Req & RFC3339 string & Event timestamp at capture point (gateway or instrumented boundary). \\
\texttt{event\_id} & Req & string & Unique identifier for this CIM event (UUID or equivalent). \\

\texttt{delegation\_id} & Req & string & Durable authorization context identifier; defines the delegation partition over events (\S2). \\
\texttt{delegation\_parent\_id} & Opt & string/null & Parent delegation identifier if authority is explicitly re-delegated (optional chaining). \\

\texttt{principal.user\_id} & Req & string & Human principal that originated the delegated authority (stable ID). \\
\texttt{agent.agent\_id} & Req & string & Logical agent identity (e.g., ``research-agent''); stable across runs. \\
\texttt{agent.instance\_id} & Opt & string & Concrete agent process/session instance (restart-sensitive). \\

\texttt{tool.name} & Req & string & Tool family name (e.g., \texttt{gdocs}, \texttt{github}, \texttt{drive}, \texttt{slack}). \\
\texttt{tool.operation} & Opt & string & Tool-native operation name (e.g., API endpoint, RPC method). \\

\texttt{action.semantic} & Req & enum string & Tier-1 normalized semantic action (\S\ref{app:action-taxonomy}), e.g., \texttt{read}, \texttt{write}, \texttt{share}, \texttt{invoke}. \\
\texttt{action.outcome} & Opt & enum string & \{\texttt{success}, \texttt{failure}, \texttt{unknown}\}. \\

\texttt{action.tier2.direction} & Opt & enum string &
Tier-2 qualifier for boundary crossing: \{\texttt{internal}, \texttt{external}, \texttt{unknown}\}. Used in \S4 queries (e.g., externalization). \\

\texttt{action.tier2.principal \_scope} & Opt & enum string &
Tier-2 qualifier for the target principal category, e.g., \{\texttt{self}, \texttt{org}, \texttt{external}, \texttt{unknown}\}. Used for ``shared with external''-style analyses in \S4. \\

\texttt{resource.type} & Opt & enum string &
Resource category (e.g., \texttt{document}, \texttt{repo}, \texttt{channel}, \texttt{ticket}). (If you use \texttt{resource\_class} in \S4, treat this field as the same concept.) \\
\texttt{resource.id} & Opt & string & Tool-scoped identifier (e.g., doc ID, repo path). \\
\texttt{resource.uri} & Opt & string & Canonical URI/locator if available (may be redacted). \\

\texttt{resource.sensitivity \_label} & Opt & enum string &
Tier-2 qualifier for sensitivity (e.g., \texttt{public}, \texttt{internal}, \texttt{confidential}, \texttt{restricted}, \texttt{unknown}). Used in \S4 blast-radius style summaries. \\

\texttt{target.principal\_id} & Opt & string & Recipient/target principal for sharing/invites (e.g., external email). \\

\texttt{trace.trace\_id} & Opt & string & Tracing identifier if present; encodes control-flow structure. \\
\texttt{trace.span\_id} & Opt & string & Span identifier if present. \\
\texttt{trace.parent\_span\_id} & Opt & string/null & Parent span identifier if present. \\

\texttt{workflow.workflow\_id} & Opt & string & Optional higher-level workflow correlation ID (may group multiple delegations). \\
\texttt{labels} & Opt & map[string]string & Free-form labels for deployment-specific tagging. \\
\bottomrule
\end{tabular}
\end{table*}

\subsection{Semantic Action Normalization Taxonomy}
\label{app:action-taxonomy}

CIM defines a  cross-tool action vocabulary to enable consistent querying across heterogeneous systems. The taxonomy is intentionally compact and may be extended. We group actions by intent:

\paragraph{Read and discovery.}
\begin{itemize}
    \item \texttt{read}: retrieve content or metadata (e.g., file read, document get).
    \item \texttt{search}: query or list resources (e.g., search documents, list repos).
    \item \texttt{enumerate}: enumerate principals or permissions (e.g., list members, list ACLs).
\end{itemize}

\paragraph{Write and modification.}
\begin{itemize}
    \item \texttt{write}: create or modify a resource (e.g., edit document, commit file).
    \item \texttt{delete}: remove a resource.
    \item \texttt{move}: change resource location/namespace (e.g., move file/folder).
\end{itemize}

\paragraph{Sharing and authorization-relevant actions.}
\begin{itemize}
    \item \texttt{share}: grant access or distribute a resource (e.g., add collaborator, create share link).
    \item \texttt{revoke}: remove access (e.g., remove collaborator, disable link).
    \item \texttt{invite}: invite a principal to a workspace/channel/project.
\end{itemize}

\paragraph{Execution and delegation.}
\begin{itemize}
    \item \texttt{invoke}: invoke a tool operation that triggers execution (e.g., run job, open PR workflow).
    \item \texttt{delegate}: create or bind a delegation context (optional explicit event).
    \item \texttt{spawn}: create a child agent instance (optional explicit event).
\end{itemize}

\paragraph{Mapping guidance.}
Normalization is performed using tool-native metadata (operation/method name, endpoint class, and resource type), not payload content. For example, operations named \texttt{get}, \texttt{download}, \texttt{readFile} map to \texttt{read}; operations \texttt{update}, \texttt{commit}, \texttt{edit} map to \texttt{write}; operations \texttt{addMember}, \texttt{setPermission}, \texttt{createLink} map to \texttt{share}. When an operation is ambiguous, implementations may emit both \texttt{tool.operation} and \texttt{action.semantic} and record \texttt{action.outcome}=\texttt{unknown} until disambiguated.

\pagebreak

\subsection{Worked Example: End-to-End Workflow and CIM Events}

To make CIM concrete, we provide a minimal end-to-end workflow and example CIM events emitted by the gateway. The workflow is representative of common enterprise agentic tasks and is intentionally content-minimal: events capture actors, actions, resources, and structure, but not prompts or payload bodies.

A user delegates an agent to: \emph{``Summarize Q4 sales performance and create a memo.''}
The agent performs four steps spanning heterogeneous systems:
(i) reads an internal spreadsheet, (ii) queries a CRM, (iii) writes a memo document, and (iv) shares the memo to a restricted channel. All steps occur under a single delegated authorization context.
Listing~\ref{lst:cim-workflow-events} shows four example CIM events for this workflow. Note that events can span multiple trace components; they remain bound by \texttt{delegation\_id}. For simplicity, we show a single agent instance (\texttt{agent\_id}, \texttt{instance\_id}) executing all steps under one delegation; in practice, multi-agent workflows may emit the same \texttt{delegation\_id} across multiple agent instances (or descendant delegations), while traces may fragment across tool adapters and runtimes.

\begin{lstlisting}[language={},caption={Example CIM events for a single delegated workflow (four steps across tools).},label={lst:cim-workflow-events}]
[
  {
    "time": "2026-02-13T10:00:01.120Z",
    "event_id": "e1",
    "delegation_id": "del_7b3a9c1e",
    "principal": { "user_id": "user_12345" },
    "agent": { "agent_id": "analysis-agent", "instance_id": "agentinst_a9" },
    "tool": { "name": "drive", "operation": "files.get" },
    "action": { "semantic": "read", "outcome": "success" },
    "resource": { "type": "spreadsheet", "id": "file_sales_q4", "uri": "drive://file_sales_q4" },
    "trace": { "trace_id": "T1", "span_id": "S1", "parent_span_id": null }
  },
  {
    "time": "2026-02-13T10:00:03.410Z",
    "event_id": "e2",
    "delegation_id": "del_7b3a9c1e",
    "principal": { "user_id": "user_12345" },
    "agent": { "agent_id": "analysis-agent", "instance_id": "agentinst_a9" },
    "tool": { "name": "crm", "operation": "accounts.search" },
    "action": { "semantic": "search", "outcome": "success" },
    "resource": { "type": "account", "id": "acct_query_q4", "uri": "crm://accounts?q=q4" },
    "trace": { "trace_id": "T2", "span_id": "S7", "parent_span_id": null }
  },
  {
    "time": "2026-02-13T10:00:07.905Z",
    "event_id": "e3",
    "delegation_id": "del_7b3a9c1e",
    "principal": { "user_id": "user_12345" },
    "agent": { "agent_id": "analysis-agent", "instance_id": "agentinst_a9" },
    "tool": { "name": "gdocs", "operation": "documents.create" },
    "action": { "semantic": "write", "outcome": "success" },
    "resource": { "type": "document", "id": "doc_q4_memo", "uri": "gdocs://doc_q4_memo" },
    "trace": { "trace_id": "T3", "span_id": "S2", "parent_span_id": "S1" }
  },
  {
    "time": "2026-02-13T10:00:10.222Z",
    "event_id": "e4",
    "delegation_id": "del_7b3a9c1e",
    "principal": { "user_id": "user_12345" },
    "agent": { "agent_id": "analysis-agent", "instance_id": "agentinst_a9" },
    "tool": { "name": "slack", "operation": "chat.postMessage" },
    "action": { "semantic": "share", "outcome": "success" },
    "resource": { "type": "message", "id": "msg_8891", "uri": "slack://channel/finance-q4/msg_8891" },
    "target": { "principal_id": "slack://channel/finance-q4" },
    "trace": { "trace_id": "T4", "span_id": "S9", "parent_span_id": null }
  }
]
\end{lstlisting}

\section{Gateway Architecture and Deployment Contract}
\label{app:gateway}

\noindent\textbf{Motivation and role.}
Relying solely on instrumentation inside agent runtimes is brittle in heterogeneous deployments and weak under partial trust, where agent implementations vary, may be misconfigured, or may be adversarially modified. We therefore introduce a \emph{gateway} as a uniform observability boundary between agents and tools. By mediating tool invocations, the gateway (i) binds each observed action to a stable \texttt{delegation\_id}, (ii) emits delegation lineage edges at delegation boundaries, and (iii) applies deterministic cross-tool normalization into CIM events. The gateway is observational rather than preventative: it captures structural metadata required for delegation-aware reconstruction and does not require prompt, payload, or model-internal reasoning inspection.

\subsection{Design Principles}
\label{app:gateway:principles}
\noindent\textbf{Progressive enhancement.}
Basic audit capture (timestamps, tool endpoint, operation, outcome) must succeed even when higher-level enrichment fails, ensuring graceful degradation rather than silent data loss.

\noindent\textbf{Determinism and semantic stability.}
Normalization from tool-specific operations to CIM semantics is rule-based and deterministic to support longitudinal comparability and reproducible analyses.

\noindent\textbf{Defense in depth.}
The architecture assumes agent code may be untrusted. Gateways may be isolated from agent runtimes, enforce mandatory routing within a declared scope, and apply tamper-evident protections (e.g., event signing) to emitted telemetry to strengthen audit integrity.

\subsection{Gateway Contract: Four Logical Functions}
\label{app:gateway:contract}
Each gateway instance implements four logical functions, independent of any particular agent framework:

\noindent\textbf{(1) Interception.}
Capture agent-initiated requests to external systems (e.g., proxy/sidecar, SDK wrapper, or in-service hooks) and extract protocol-level metadata such as target system, operation, stable resource identifiers when available, and outcome.

\noindent\textbf{(2) Context binding.}
Attach delegation and actor context to each intercepted action by minting or validating \texttt{delegation\_id} and, when applicable, linking to \texttt{delegation\_parent\_id} to form a delegation lineage.

\noindent\textbf{(3) Semantic normalization.}
Deterministically map tool-specific operations into CIM-normalized action types and resource classes (e.g., \texttt{read}, \texttt{write}, \texttt{share}, \texttt{invoke}; \texttt{doc}, \texttt{repo}, \texttt{ticket}), without content inspection or learned inference.

\noindent\textbf{(4) Event emission.}
Emit a CIM-compliant record containing the event envelope (timestamp, tool endpoint), actor and delegation identifiers, resource attribution, outcome, and (when propagated) workflow correlation fields such as trace/span identifiers.

\subsection{Field Provenance and Integrity Semantics}
\label{app:gateway:provenance}
To make assumptions explicit, CIM fields are partitioned into two provenance classes.

\noindent\textbf{Gateway-bound fields.}
These include authenticated principal identity, observed agent instance identity, timestamp, tool endpoint, and delegation binding (\texttt{delegation\_id} minted/validated at the boundary plus lineage linkage). These fields are authoritative within the declared scope because they are derived from gateway observation and authenticated channels.

\noindent\textbf{Agent-provided fields.}
These are treated as annotations (e.g., task labels, optional correlation IDs when originating from the runtime). When agent-provided context conflicts with gateway-bound facts (e.g., principal binding, tool endpoint), the gateway-bound value is authoritative and the inconsistency may be recorded as an integrity signal for downstream analysis.

\subsection{Layered Gateway Design}
\label{app:gateway:layers}
We realize the contract via four composable layers:

\noindent\textbf{Protocol parsing layer.}
Understands framework- and protocol-specific formats (e.g., MCP, REST, SDKs) and extracts tool names, operations, request identifiers, and outcomes without access to prompts or model state.

\noindent\textbf{Semantic enrichment layer.}
Applies deterministic rules to map raw operations to normalized CIM semantics, including classification of action types, identification of resource classes and identifiers, and categorization of tools.

\noindent\textbf{CIM generation layer.}
Transforms enriched metadata into CIM-compliant events, assigning stable identifiers for agent identity and delegation context, and emitting lineage edges and workflow correlation fields when available.

\noindent\textbf{Log emission layer.}
Exports CIM events to downstream sinks (e.g., log analytics platforms, message queues, object storage) with buffering, retry, and backpressure handling, decoupling capture from storage/analysis.

\subsection{Deployment Patterns and Coverage Model}
\label{app:gateway:deployment}
The gateway supports multiple configurations, which may be composed within a deployment:

\noindent\textbf{Sidecar proxy.}
A sidecar intercepts agent traffic at the network boundary. When coupled with network policy, it can enforce mandatory routing and provide strong bypass resistance.

\noindent\textbf{SDK wrapper.}
Audited SDKs intercept tool invocations within the agent runtime, providing low-friction integration and low latency but typically offering best-effort coverage due to coupling with agent implementations.

\noindent\textbf{In-service instrumentation.}
Services under organizational control emit CIM events directly as part of request handling. This provides high semantic fidelity and is often combined with sidecar proxies for broader coverage.

\noindent\textbf{Mandatory vs. best-effort routing.}
Deployments may enforce mandatory routing within a declared scope, making absence of gateway telemetry meaningful within that scope, or run in best-effort mode where CIM enables reconstruction over covered tools but completeness claims require explicit coverage accounting.

\noindent\textbf{Coverage accounting.}
In partial deployments where only some tools are mediated, CIM supports delegation-scoped analyses over the covered subset (e.g., blast radius, lineage-based attribution, drift), while uncovered channels must be treated as unknown. Global accountability statements therefore require either mandatory routing for the declared scope or explicit mechanisms to detect missing coverage.

\section{Threat Model, Scope, and Bypass Detection}
\label{app:gateway:threat}

\noindent\textbf{Scope and positioning.}
CIM is an observability substrate, not an enforcement or prevention 
mechanism. The gateway raises the ceiling on what can be observed 
and reconstructed post hoc; it does not block misuse at execution 
time. This framing is intentional: CIM's value is structural 
accountability---making composite behaviors visible and attributable 
after the fact---rather than real-time policy enforcement, which 
is the domain of access-control and authorization systems . As with any audit-based system, 
evasion remains possible, but evasion leaves a detectable absence 
rather than a false positive.

\noindent\textbf{Attacker model.}
We consider a motivated insider: a principal who has legitimate 
access to the agentic system and delegated authority, but who 
attempts to use that authority in ways that evade delegation-aware 
observability. This covers both (a)~\textit{malicious insiders} 
who deliberately structure delegations to avoid detection, and 
(b)~\textit{negligent insiders} whose broadly scoped delegations 
produce harmful composite behaviors without deliberate evasion. 
We do not model external adversaries who have compromised the 
gateway itself; tamper-evident storage and gateway hardening 
(orthogonal to CIM semantics) address that surface.

\noindent\textbf{Evasion class 1: Out-of-band channels.}
An insider may act through tools or services that are not 
instrumented by the gateway---for example, directly calling an 
API outside the declared tool inventory, using a personal device, 
or routing through a shadow integration. CIM provides no direct 
signal for actions outside its coverage boundary.

\textit{CIM's residual utility.} CIM's explicit absence semantics 
make coverage gaps detectable: if a tool appears in a declared 
inventory but emits no CIM events during a delegation window, 
the absence is itself a signal. In mandatory-routing deployments, 
any uninstrumented egress violates the routing policy and can 
trigger alerts independently of CIM. In best-effort deployments, 
investigators can identify which tools were covered and which were 
not, bounding the unobservable surface explicitly rather than 
assuming complete coverage. The net effect is that out-of-band 
evasion is possible but not silent: it requires the insider to 
actively avoid instrumented channels, which is itself a 
detectable behavioral pattern when compared against baseline 
tool usage.

\noindent\textbf{Evasion class 2: Delegation fragmentation.}
A sophisticated insider may issue many narrow, short-lived 
delegations rather than one broad one, distributing actions 
across delegation contexts to prevent any single lineage-scoped 
query from revealing the composite behavior. For example, 
instead of one delegation $d$ that reads ten sensitive documents 
and shares them externally, the insider issues ten delegations 
$d_1, \ldots, d_{10}$, each reading one document, and a separate 
delegation $d_{11}$ for the share action.

\textit{CIM's residual utility.} Lineage closure (R3) still 
attributes all child delegations to their originating principal 
via \texttt{delegation\_parent\_id}, so per-principal aggregation 
remains intact even when per-delegation blast radius appears 
narrow. CIM's normalized action vocabulary (R4) further supports 
cross-delegation queries: an analyst can query all \texttt{share} 
actions with \texttt{direction=external} attributed to a principal 
$p$ within a time window, regardless of how many delegations 
were used. Fragmentation therefore shifts the evasion burden 
from hiding individual actions to hiding the principal relationship, 
which CIM's principal binding (R1, I5) is specifically designed 
to preserve. Fragmentation also leaves a structural signal: 
an unusually high rate of short-lived delegations from a single 
principal is itself an anomaly detectable from CIM telemetry 
without payload inspection.

\noindent\textbf{Evasion class 3: Identifier spoofing.}
If \texttt{delegation\_id} values are accepted as 
agent-asserted free-form strings rather than gateway-minted 
and principal-bound tokens, a compromised or malicious agent 
component could inject false delegation identifiers---attributing 
its actions to a benign delegation or to another principal's 
context. 

\textit{CIM's residual utility.} Invariant I5 (binding integrity, 
\S\ref{sec:cim}) directly addresses this: \texttt{delegation\_id} 
minting and validation occur at the gateway boundary, binding 
each identifier to an authenticated principal, an issuance 
timestamp, and a session context. Agent-provided annotations 
are treated as untrusted hints and are not used as authoritative 
delegation bindings, field provenance). 
Consistency check (i) from Appendix~D further 
detects spoofing attempts: if a single \texttt{delegation\_id} 
is observed binding to multiple distinct \texttt{principal.user\_id} 
values, this is flagged as an integrity violation. Spoofing 
is therefore defeated within the instrumented channel as long 
as gateway-minted bindings are used; it remains a risk only 
if the gateway itself is compromised, which is addressed by 
tamper-evident storage and gateway isolation.

\noindent\textbf{Residual utility under partial evasion.}
A key property of CIM is that it degrades gracefully: partial 
evasion narrows the observable surface rather than collapsing 
observability entirely. Formally, let $T \subseteq \text{Tools}$ 
be the set of instrumented tools. CIM guarantees 
delegation-scoped reconstruction for all actions traversing 
$T$; actions on $\text{Tools} \setminus T$ are unobservable 
but their absence is detectable relative to the declared 
inventory. An insider who evades CIM on $k$ tools still 
leaves a fully reconstructable record for all remaining 
tool interactions, and the evasion itself---the gap between 
expected and observed tool coverage---is a structural signal. 
This is strictly better than baseline telemetry, under which 
the same partial record cannot be reliably attributed to a 
delegation context even for the covered tools.

\noindent\textbf{Assumptions within declared scope.}
Security-grade conclusions rely on (i) an authenticated channel 
that binds principals and agent instances at the gateway boundary, 
(ii) either mandatory routing or a mechanism to detect missing 
coverage, and (iii) gateway-minted delegation identifiers with 
principal binding (I5). Tamper-evident storage and gateway 
isolation strengthen evidentiary guarantees against compromised 
infrastructure but are orthogonal to CIM semantics and are 
expected to be provided by the deployment environment.

\begin{table*}[t]
\centering
\small
\setlength{\tabcolsep}{5pt}
\renewcommand{\arraystretch}{1.22}
\begin{tabular}{p{0.25\linewidth} p{0.36\linewidth} p{0.36\linewidth}}
\hline
\textbf{Threat} &
\textbf{CIM query (delegation-scoped)} &
\textbf{Without CIM (baseline correlation)} \\
\hline

\textbf{Sensitive cross-tool aggregation:}
one delegated task touches multiple high-sensitivity resources across tools (e.g., HR + finance + CRM). &
Filter $\texttt{delegation\_id}=d$ (or $d \cup \mathrm{Desc}(d)$); group by \texttt{delegation\_id}; compute $\mathrm{count\_distinct}(\texttt{resource\_class})$ and $\max(\texttt{sensitivity\_label})$; alert if $\max(\texttt{sensitivity\_label}) \ge \texttt{High}$ and $\mathrm{count\_distinct}(\texttt{resource\_class}) \ge k$. &
Join tool logs (Docs/Drive/CRM/etc.); map tool-specific IDs/actions/resources to a common schema; join sensitivity labels per tool (lookups); unify identities (directory join); correlate by (principal, agent instance) + time window $\Delta t$; dedupe retries/backtracking; split/merge concurrent tasks; alert on cross-tool co-occurrence (often ambiguous under overlap). \\
\hline

\textbf{High-sensitivity read $\rightarrow$ externalization:}
read \texttt{High/Restricted} then share/export/post/upload to an external principal/destination under the same delegated task. &
Within $d$ (or descendants), alert if
$\exists(\texttt{action\_tier1}=\texttt{read} \wedge \texttt{sensitivity\_label}\ge \texttt{High})$
and
$\exists(\texttt{action\_tier1}\in\{\texttt{share, export, post, upload}\} \wedge
\texttt{action.tier2.principal\_scope}=\texttt{external})$
(optionally also $\texttt{action.tier2.direction}=\texttt{egress}$). &
Detect sensitive reads per tool (tool-specific fields + label lookups); detect externalization per tool (share/export/post/upload with external targets); unify identities across systems; approximate ``same task'' via time-window correlation $\Delta t$ (and traces when available); dedupe retries; disambiguate overlapping tasks; join the two predicates within the inferred session/window (FP/FN common). \\
\hline

\textbf{Cross-tool failure burst $\rightarrow$ success (credential/permission abuse):}
many denied/failed actions across tools followed by success under the same task. &
Filter to $\texttt{delegation\_id}=d$; count failures where $\texttt{result}\in\{\texttt{denied, failed}\}$ and $\texttt{action\_tier1}=\texttt{invoke}$ with $\texttt{resource\_class}\in\{\texttt{identity},\texttt{auth}\}$; alert if $\texttt{fail\_count}\ge N$ and $\exists(\texttt{result}=\texttt{success})$ within the same delegation (optionally per target resource). &
Collect failure/deny events from heterogeneous tools; normalize \emph{result semantics} (fail/deny/success) across systems; unify identities; correlate attempts into campaigns via time windows $\Delta t$ (or partial traces); dedupe retries; split concurrent tasks; then alert on burst$\rightarrow$success pattern (noisy when tasks overlap or traces fragment). \\
\hline

\textbf{Multi-agent descendant abuse:}
primary delegation spawns sub-agents; a descendant performs restricted access/externalization to evade ``same trace'' rules. &
Apply lineage closure: alert if
$\exists(e:\ \texttt{delegation\_id}\in \mathrm{Desc}(d)\ \wedge\
\texttt{sensitivity\_label}\ge\texttt{High}\ \wedge\
\texttt{action\_tier1}\in\{\texttt{read, share, export, upload}\})$. &
Reconstruct multi-agent relationships from traces/messages (often partial); infer parent/child linkage; join across tool logs; use time-window correlation to approximate a ``delegation tree''; dedupe retries; apply heuristics to attribute descendant actions to the originating task (high ambiguity when delegation boundaries are not explicit). \\
\hline
\end{tabular}

\vspace{3pt}
\caption{Worked threat-detection examples showing how CIM replaces heuristic cross-tool joins and time-window correlation with delegation-scoped predicates. Thresholds (e.g., $k,N,\Delta t$) are deployment-configurable; shown forms are representative of the benchmarked query shapes.}
\label{tab:threat_query_examples}
\end{table*}

\section{Threat-Detection Query Examples (CIM vs.\ Baselines)}
\label{app:threat_queries}

Table~\ref{tab:threat_query_examples} highlights the query-effort gap.
The key point is operational rather than micro-performance: in many real deployments, the dominant cost
is not millisecond-level query latency but the effort required to \emph{author, validate, and maintain}
queries that reconstruct delegation-scoped behavior from heterogeneous event sources.
When delegation membership is not directly observable, even simple questions typically expand into
pipelines of joins, per-source normalization/lookups, identity reconciliation, time-window correlation,
and heuristic post-processing for retries and overlap. These steps are expensive to build correctly,
brittle under concurrency, and costly to maintain as tool schemas evolve.
CIM makes delegation membership explicit at execution time for covered events, converting
correlation-heavy pipelines into delegation-scoped predicates and thereby reducing query construction
and maintenance burden.

\paragraph{How to read the table.}
Each row is a representative delegation-scoped question shown in two forms.
The \textbf{CIM query} column expresses the question as a predicate over normalized fields (e.g.,
\texttt{delegation\_id}, \texttt{action\_tier1}, \texttt{resource\_class}, \texttt{sensitivity\_label},
and externalization attributes such as \texttt{action.tier2.principal\_scope} and
\texttt{action.tier2.direction}). These queries typically require filtering on a delegation (and
optionally its descendant closure), followed by a small aggregation and predicate.
The \textbf{Without CIM} column summarizes a typical baseline \emph{operator plan}---not tied to any
specific query language---that must be implemented when delegation membership is inferred indirectly.

\paragraph{What ``operations'' mean (link to Table~7).}
The operation counts in Table~7 are intended as a proxy for query construction and maintenance effort:
they correspond to major steps in the baseline plans, including (i) cross-source joins,
(ii) schema mappings and per-source normalization/lookups (e.g., metadata labels and resource
attributes), (iii) identity-unification joins, (iv) windowed correlation ($\Delta t$) to approximate
task boundaries, and (v) heuristic post-processing to handle retries/backtracking and overlap
(split/merge decisions). Beyond runtime, each operation is additional logic that must be written,
tested, tuned, and kept consistent as sources change. CIM reduces or eliminates these reconstruction
steps by providing a stable delegation-scoped join key and lineage closure.

\paragraph{Thresholds and deployment policy.}
Thresholds in Table~\ref{tab:threat_query_examples} (e.g.,
$\mathrm{count\_distinct}(\cdot)\ge k$, $\texttt{fail\_count}\ge N$, and correlation window $\Delta t$)
are deployment-configurable and are shown only to illustrate query \emph{shape}.

\paragraph{Coverage and limitations.}
These examples assume events are captured through the CIM emission path (e.g., a gateway/middleware).
Out-of-band actions that bypass this path are not directly observable.
CIM's contribution is that, for \emph{covered} events, delegation-scoped reconstruction is well-posed
and does not require heuristic task inference, reducing correlation complexity---and thus query
construction burden---in downstream analyses.

\end{document}